\documentclass[conference]{IEEEtran}
\IEEEoverridecommandlockouts

\usepackage{cite}
\usepackage{mathtools}
\usepackage{amsmath,amssymb,amsfonts}
\usepackage{amsthm}
\usepackage{algorithm}
\usepackage{algorithmic}
\newtheorem{proposition}{Proposition}
\usepackage{graphicx}
\usepackage{booktabs}
\usepackage{subcaption}
\usepackage{url}
\usepackage[table]{xcolor}
\usepackage{textcomp}
\usepackage{comment}
\usepackage{booktabs}
\usepackage{tabularx}
\usepackage{array}
\usepackage{enumitem}
\usepackage{url}
\usepackage[most]{tcolorbox}
\usepackage[hidelinks]{hyperref}
\newcolumntype{L}[1]{>{\raggedright\arraybackslash}p{#1}}
\newcolumntype{Y}{>{\raggedright\arraybackslash}X}

\def\BibTeX{{\rm B\kern-.05em{\sc i\kern-.025em b}\kern-.08em
    T\kern-.1667em\lower.7ex\hbox{E}\kern-.125emX}}

\newcommand{\Mt}{M_t}
\newcommand{\Wt}{W_t}
\newcommand{\Ut}{U_t}
\newcommand{\Dt}{D_t}
\newcommand{\At}{A_t}
\newcommand{\Bt}{B_t}
\newcommand{\RT}{R_T}
\newcommand{\RD}{R_D}

\newcommand{\code}[1]{\texttt{\small #1}}

\newcolumntype{C}[1]{>{\centering\arraybackslash}p{#1}}

\begin{document}

\title{TrackShield: An Efficient Runtime Defense Against Availability Attacks on High-Performance Real-Time Multi-Object Tracking
\thanks{This work was partially supported by the Advanced Summer
Program in Research (ASPiRE) at Hofstra University, and the Swedish Research Council under Grant No. 2023-04485.}
}

\author{
\IEEEauthorblockN{
Zonghua Gu\textsuperscript{1},
Julian Singh-Smith\textsuperscript{1},
Junlin Liao\textsuperscript{1},
and Di Liu\textsuperscript{2}
}
\IEEEauthorblockA{
\textsuperscript{1}Hofstra University, Hempstead, NY, USA\\
\textsuperscript{2}Norwegian University of Science and Technology,
Trondheim, Norway
}
}

\maketitle
\begin{abstract}
High-performance camera-based perception pipelines must sustain per-frame real-time budgets on resource-constrained embedded GPU platforms. Availability attacks such as SlowTrack inflate the persistent state of multi-object trackers: data-association workload accumulates across frames, so per-frame latency increases and drives the pipeline past its deadline. Detector- and NMS-layer defenses do not directly control this persistent tracker-side resource. We present TrackShield, a lightweight runtime resource-management layer that protects the tracker at the point where attack-induced workload accumulates. TrackShield combines low-overhead anomaly monitoring with staged birth capping, pruning of young or weakly matched tracks, and per-stage association-input admission control. Its hard admission guard provides a deterministic bound on admitted association workload, independently of the heuristic monitor. On a Jetson AGX Orin, TrackShield reduces attacked tracking-stage latency from $139.90$ to $6.27$\,ms and whole-pipeline P99 latency from $190.9$ to $48.8$\,ms, at a benign MOTA cost of $1.06$ points. As an availability defense, TrackShield focuses on keeping the tracker responsive under attack instead of recovering the full tracking accuracy degraded by adversarial detections.
\end{abstract}

\section{Introduction}
\label{sec:intro}

\begin{figure*}[!t]
\centering
\includegraphics[width=0.92\textwidth]{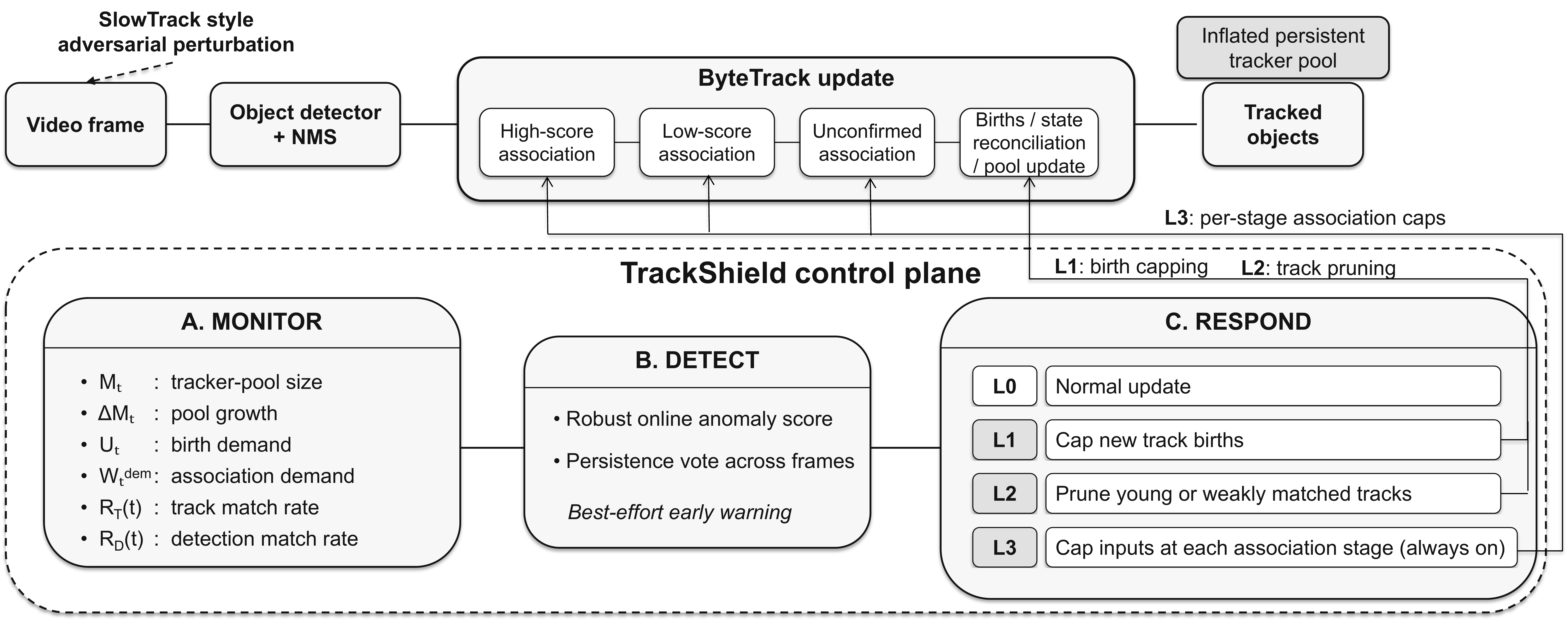}
\caption{\textbf{TrackShield overview.} SlowTrack-style perturbations
inflate ByteTrack's persistent tracker pool and association demand.
TrackShield monitors tracker-state signals (counts only, not timing)
using robust EWMA/MAD scores and a $q$-of-$w$ persistence vote,
escalating the heuristic response from normal operation to L1 birth
capping and L2 pruning of young or weakly matched tracks. The always-on L3 hard admission guard caps the track and
detection candidates admitted to every protected association stage.}
\label{fig:overview}
\end{figure*}

Camera-based perception in autonomous driving must complete within a real-time
per-frame budget, commonly on the order of 33--100\,ms in the evaluated
frame-rate range (the specific value is workload- and architecture-dependent), to
feed downstream planning and control. Perception is conventionally organized as a
\textbf{tracking-by-detection} pipeline: a detector (e.g., YOLOX~\cite{yolox}) produces bounding boxes per frame, and a multi-object tracker (e.g., SORT~\cite{sort},
ByteTrack~\cite{bytetrack}, BoT-SORT~\cite{botsort}) associates those boxes across frames into
persistent tracks using a motion model (Kalman filter) and an assignment step
(Hungarian / greedy matching over an IoU or appearance cost matrix).

Most adversarial-ML research on this pipeline has targeted
\textbf{integrity}, causing objects to disappear, be misclassified, or have their identities swapped~\cite{trackerhijacking}. A distinct class of attack
targets \textbf{availability}: the attacker's optimization target is
computational cost, forcing a real-time deadline
miss~\cite{gu2026deep}. The two are not exclusive. The
adversarial detections produced by an availability attack may also degrade
tracking integrity, and in our measurements they do so severely
(Section~\ref{sec:accuracy}); we measure the two effects separately. Availability attacks threaten the timeliness of downstream perception-dependent control.

SlowTrack~\cite{slowtrack} exposes a gap in detector-side defenses:
an adversary can inflate persistent tracker state across frames even
when the number of detections retained in each frame is capped.
Existing approaches primarily transform or harden detector inputs and
outputs, whereas the exploited workload is formed inside the tracker
update. TrackShield addresses this gap at runtime by monitoring and
controlling tracker-pool growth and admitted association work.

This paper makes three contributions:
\begin{enumerate}
\setlength{\itemsep}{1pt}
\setlength{\parskip}{0pt}

\item \textbf{A runtime tracking-layer defense.}
TrackShield acts inside tracker management, where SlowTrack-class attacks
accumulate persistent state, and combines staged birth capping, pruning of young or weakly matched
tracks, and association-input capping without retraining.

\item \textbf{Hardware-agnostic signal definitions.} TrackShield monitors six
tracker-internal quantities: pool size $M_t$, pool growth $\Delta M_t$,
birth demand $U_t$, the coarse association-demand proxy
$W_t^{\mathrm{dem}}=M_t|\mathcal D_t|$, and the track- and
detection-association ratios $R_T(t)$ and $R_D(t)$, where
$\mathcal D_t$ is the full post-threshold detection set entering the
tracker.

\item \textbf{A deterministic admitted-workload bound.}
The hard admission guard caps the track and detection candidates entering each
protected association stage, yielding the per-stage dimension bound of
Proposition~\ref{prop:cap}, independently of the heuristic monitor.
\end{enumerate}

\section{Related Work}
\label{sec:related}

\paragraph*{Availability attacks on detectors and pipelines}
\looseness=-1 Availability attacks increase inference-time computation rather than
directly corrupting model outputs; a recent cross-domain survey reviews
such latency attacks and their defenses~\cite{gu2026deep}. In camera perception, attacks such
as Overload generate dense candidate boxes to amplify detector
post-processing and NMS cost~\cite{overload}. Corresponding defenses
operate primarily at the input or detector layer: Schoof \emph{et al.} report that dynamic random $y$-axis resizing reduces
NMS execution time by up to $90.18\%$ on an attacked YOLOv5-L model
\cite{beyondphantomsponges}; Underload hardens the detector through
adversarial training; and the EVADE evaluation framework tests defenses
including a fixed post-NMS detection cap~\cite{beyondphantomsponges,underload,evade}.
EVADE's own findings also question the practical transfer of several NMS
latency attacks across hardware and model families~\cite{evade}, further
motivating evaluation at the actual resource boundary a given attack exploits (for TrackShield, the tracker pool). These
approaches can limit detector-side or per-frame work, but they do not manage
state that persists in a downstream tracker.

\paragraph*{SlowTrack and tracker-pool inflation}
\looseness=-1 SlowTrack extends availability attacks to the full
detection--tracking pipeline by creating detections that initialize
phantom tracks and keep them alive across frames~\cite{slowtrack}.
With retention horizon \(R\) and per-frame detection cap
\(\lvert D\rvert_{\max}\), its analysis gives a tracker pool of
\(m=R\lvert D\rvert_{\max}\) and association complexity
\(O(R\lvert D\rvert_{\max}^{2})\). SlowTrack remains effective under
the top-500 detection cap in its ByteTrack pipeline because
that cap bounds current-frame boxes, not the cross-frame tracker pool.
TrackFlood~\cite{gu2026trackflood} shows that NMS-free detectors relocate
rather than eliminate this attack surface, shifting the availability
bottleneck to downstream tracking. Its bounded-admission layer limits the
input at the detector--tracker boundary, whereas TrackShield caps both
the track and detection dimensions of each association stage and controls
the persistent state that accumulates across frames. Boundary admission
also bounds per-detection work before association, such as ReID
(Section~\ref{sec:crosstracker}), which in-tracker caps cannot; the two defenses are complementary.
Related MOT attacks target identities or trajectories rather than
computational availability~\cite{trackerhijacking}.

\paragraph*{Runtime resource controls and MOT robustness}
\looseness=-1 Runtime controls are effective only when placed at the resource
boundary they are intended to constrain. Constant-time NMS is discussed as a
mitigation direction, though the released code for that work
implements the timing attack rather than a constant-time defense~\cite{biton};
it and fixed detection caps bound aspects of detector post-processing~\cite{biton,evade}. SlowPerception~\cite{slowperception}
suggests, but does not evaluate, limiting the tracking-pool size as a
defense and notes the risk to benign tracking; fixed tracker caps may
discard legitimate objects in dense scenes.
TrackShield instead operates inside the tracker update without
retraining: its early responses cap births and prune young or weakly matched
tracks, while its deterministic hard admission guard bounds the dimensions of
each admitted association matrix.  GateDrain~\cite{gu2026gatedrain} shows a similar threat in confidence-gated edge--cloud inference: bounded perturbations redirect locally answerable requests into a shared cloud queue without raising the request rate, which can sharply degrade benign tail latency near capacity, and adaptive attacks defeat the evaluated training-free preprocessing defenses. Its Bounded Escalation defense applies the same admission-side principle as TrackShield, bounding admission to the shared cloud resource through per-source budgets and protected capacity.

\begin{table}[!h]
\caption{Positioning of representative defense families. ``Work bound''
denotes the computation directly constrained by each method.}
\label{tab:positioning}
\centering
\small
\setlength{\tabcolsep}{3.2pt}
\renewcommand{\arraystretch}{1.10}
\begin{tabularx}{\columnwidth}{@{}p{2.6cm}p{1.12cm}c c X@{}}
\hline
\textbf{Method family} &
\textbf{Layer} &
\textbf{Retrain} &
\textbf{State} &
\textbf{Work bound} \\
\toprule
Input transforms /
resizing~\cite{beyondphantomsponges}
  & Input--detector & No & No & None \\
\midrule
Underload~\cite{underload}
  & Detector & Yes & No & Detector-side only \\
\hline
Detection cap~\cite{evade}
  & Detector output & No & No & Per-frame boxes \\
\hline
Constant-time NMS (proposed)~\cite{biton}
  & NMS & No & No & NMS processing \\
\hline
Pool-size threshold (suggested, not evaluated)~\cite{slowperception}
  & Tracker & No & Yes & Configured tracker resource \\
\hline
TrackShield
  & Tracker & No & Yes & Admitted association dimensions \\
\bottomrule
\end{tabularx}
\end{table}

\looseness=-1 Detector-side defenses (Table~\ref{tab:positioning}) constrain current-frame computation, whereas
TrackShield acts on persistent tracker state and bounds admitted track
association work.

\section{Background and Threat Model}
\label{sec:background}

\subsection{Tracking-by-Detection and the Tracker Pool}

Let frame $t$ yield a detection set $\Dt$ after confidence thresholding and
non-maximum suppression (NMS). A tracking-by-detection tracker maintains a
pool of tracks $\mathcal{T}_{t-1}$ from the preceding frame, with
$\Mt = |\mathcal{T}_{t-1}|$. Each frame, it (1) predicts each track's location,
typically using a Kalman filter; (2) compares the $\Mt$ predicted tracks with
the $|\Dt|$ retained detections; (3) solves an assignment and obtains a match
set $\At$; (4) initializes tracks from unmatched detections; and (5) retains
temporarily unmatched tracks in a lost-track buffer before deleting them.

\subsection{ByteTrack Data Association}
\label{sec:bytetrack}

ByteTrack~\cite{bytetrack} splits each frame's detections into
three confidence regions by two thresholds
$\tau_{\mathrm{high}}>\tau_{\mathrm{low}}$: a high band
$D_t^{\mathrm{high}}$, a middle band $D_t^{\mathrm{low}}$, and boxes below
$\tau_{\mathrm{low}}$ which are discarded. The middle band is retained because a
low score often means occluded, blurred, or truncated rather than background:
such a box is unsafe for \emph{creating} an identity but may be credible enough
to \emph{maintain} one. Each frame runs a cascade: Kalman prediction,
then association against $D_t^{\mathrm{high}}$ (high-score stage), then the unmatched tracks that are still in the Tracked state against
$D_t^{\mathrm{low}}$ (low-score stage, a recovery pass; previously Lost
tracks are re-activated only through high-score association), then unconfirmed-track resolution and births. After these stages, \emph{reconciliation} performs the tracker's
end-of-frame state bookkeeping: it applies match and birth results,
updates track states, moves unmatched tracks between the tracked and lost
sets, removes expired or duplicate entries, and constructs the persistent
state carried into frame $t+1$. Because the tracker runs \emph{three} association stages per frame,
association work is a \emph{sum} of per-stage costs, and a cap on the first
stage alone is insufficient. The low-score stage is also directly
attackable: flooding its detection set keeps phantom tracked trajectories
alive and enlarges the second-pass matrix even when the high-score stage
looks nominal.

\subsection{SlowTrack Attack Mechanism}
\label{sec:slowtrack}
SlowTrack adds an $\ell_2$-style adversarial perturbation to each frame and
optimizes score, box-area, and feature-matching losses to create numerous
retained detections, initialize new trackers, and subsequently re-find lost
trackers before they are deleted~\cite{slowtrack}. Its effect is temporal: phantom tracks accumulate over the retention horizon.

SlowTrack formalizes this mechanism through the data-association complexity
$O(m n_0)$, where $m$ is the number of tracks in the tracker pool and $n_0$ is
the number of reserved detection boxes in the current frame~\cite{slowtrack}.
Association work grows approximately with
``number of tracks $\times$ number of boxes.'' In SlowTrack's retention-horizon analysis~\cite{slowtrack}, if the detector
retains at most $|D|_{\max}$ boxes per frame and unmatched tracks persist for
up to $R$ frames, the attack-induced workload scale is
\[
    m = R|D|_{\max},
    \qquad
    O(mn_0) = O(R|D|_{\max}^{2}).
\]
The attack therefore remains effective even when a detector-side cap limits
the instantaneous detection count: the exploited resource is the tracker pool accumulated across frames. Suppose unmatched tracks expire after $R$ frames and each detection use
refreshes or creates at most one track. Then every live track was born or
matched within the last $R{+}1$ frames, so the pool holds at most $R{+}1$
times the number of detection uses per frame: $|D|_{\max}$ in stock
ByteTrack, and up to $|D|_{\max}+100$ in the evaluated tracker, whose second
stage can reuse first-stage detections (Section~\ref{sec:setup}). SlowTrack
can drive the pool far beyond real-time capacity well before this bound,
which does not hold for guard-excluded tracks on the evaluated host
(Section~\ref{sec:respond}).
Because association demand $\Wt^{\mathrm{dem}}\approx|\Dt|\Mt$ keeps growing with the pool despite a detection cap, effective mitigation must
control persistent state and association admission inside the tracker update.

\subsection{Threat Model for the Defender}
\label{sec:threat}
We assume a defender who:
\begin{itemize}[leftmargin=6pt]
  \item \textbf{Cannot} observe the perturbation, attack gradients, or
  the attacker's loss (black-box with respect to the attack).
  \item \textbf{Can} observe the tracker's internal state each frame:
  the post-threshold detection set $\Dt$, the stage-local match sets
  $\{\mathcal A_{t,s}\}_{s=1}^{S}$, the frame-entry persistent-pool size
  $\Mt$ (tracked plus lost), pre-cap birth demand $\Ut$, realized births
  $\Bt$, and each track's age, confidence, and hit streak (its matched-update count since activation or
  re-activation: the host's tracklet length, which is $0$ at activation,
  increases by one on each matched update, and resets on re-activation).
  \item Operates \textbf{at inference time only}, without retraining or
  model-weight access.
\end{itemize}

On the attacker side, we assume the SlowTrack setting: an
$\ell_2$-style perturbation is optimized against the \emph{undefended}
pipeline to increase tracker-management work. The attacker may know the
architecture and that TrackShield is deployed, and may influence monitored
tracker state through adversarial detections, but cannot bypass the guard
logic. We evaluate transferred adversarial frames; Section~\ref{sec:limitations} discusses adaptive attacks.

\begin{tcolorbox}[breakable,colback=gray!6,colframe=gray!45,boxrule=0.4pt,
arc=1.5pt,left=4pt,right=4pt,top=3pt,bottom=3pt,
fonttitle=\bfseries\small,title={Intuition: protect the queue, not only the door}]
\small
An MOT system has two stages: a detector that draws bounding boxes around objects in each frame, and a tracker that carries identities across frames. The tracker keeps a pool of tracks, including recently lost ones it hopes to re-find. Think of the tracker as a supermarket checkout queue: the detector controls
how many customers arrive per minute (detections per frame), while the tracker
pool is the backlog of customers still waiting (tracks carried across frames). A per-frame detection cap limits arrivals but not the accumulated queue.
SlowTrack exploits this gap by sending in phantom customers who never check out (tracks that persist across frames), growing the pool and hence the
association work, which is roughly proportional to the size of the
association matrix. TrackShield instead monitors the queue itself (its size, growth, birth pressure, and association rates) and progressively caps new
births, prunes young or weakly matched tracks, and, when necessary, caps the tracks
and detections admitted to each association stage. These controls use tracker-internal counts and ratios only, requiring no GPU timing, retraining, or additional neural network.
\end{tcolorbox}

\section{TrackShield Design}
\label{sec:design}
TrackShield consists of a set of hook functions (Fig.~\ref{fig:overview}) executed inside the existing tracker update, using 
variables that the tracker already computes. 

\subsection{Monitored Signals}
\label{sec:signals}
We form a six-dimensional feature vector from counters already
available in ByteTrack:
\begingroup
\begin{equation}
x_t=\big[\Mt,\Delta\Mt,U_t,\Wt^{\mathrm{dem}},\RT,\RD\big],
\label{eq:features}
\end{equation}
\endgroup
where $\Mt$ is the frame-entry persistent pool (tracked plus lost),
$\Delta\Mt=\Mt-M_{t-1}$, $\Ut$ is pre-cap birth demand (the number of unmatched detections presented to the birth step before ByteTrack's activation-score filter and L1 capping, so it can exceed the realized births $\Bt$), and
$\Wt^{\mathrm{dem}}=\Mt|\Dt|$ is a coarse association-demand proxy,
where $\Dt$ is the full post-threshold detection set entering the
tracker before its high- and low-score partition. This frame-level quantity is a coarse demand proxy rather than the realized association work, which depends on the candidate dimensions at
the individual association stages. All six signals are counts or ratios
rather than wall-clock times: a latency-based trigger would need re-tuning
for each device and would fire under thermal throttling or contention from
other processes.

Let $\widetilde{\mathcal T}_{t,s}$ and
$\widetilde{\mathcal D}_{t,s}$ denote the track and detection candidates
admitted to protected stage $s$, immediately before its cost matrix is
constructed, and let $A_{t,s}$ be that stage's matches. We define
\begin{equation}%
\RT(t)=
\frac{\sum_s|A_{t,s}|}
     {\sum_s|\widetilde{\mathcal T}_{t,s}|},
\qquad
\RD(t)=
\frac{\sum_s|A_{t,s}|}
     {\sum_s|\widetilde{\mathcal D}_{t,s}|}.
\label{eq:ratios}
\end{equation}
These stage-local, post-cap definitions guarantee $\RT,\RD\in[0,1]$. 

The ratios are valid only when candidates exist on both sides:
\begin{equation}\small
a_t=\mathbf{1}\!\left[
\sum_s|\widetilde{\mathcal T}_{t,s}|>0
\land
\sum_s|\widetilde{\mathcal D}_{t,s}|>0
\right].
\label{eq:assocvalid}
\end{equation}
This convention could mark a frame valid when no single stage has both
tracks and detections; no logged run produced such a frame.
When $a_t=0$, both ratios are assigned the neutral value $1$ and excluded
from scoring and baseline adaptation. The remaining signals capture the
attack's main effects: pool state and growth, birth pressure, and
association demand. Trusted benign data initialize the monitor before
deployment. %

\noindent\textbf{Toy example.}
Suppose one frame has three protected association stages with admitted
track/detection counts $(|\widetilde{\mathcal T}_{t,s}|, |\widetilde{\mathcal D}_{t,s}|) = ((4,5), \quad (2,3),\quad (1,2))$, and the stages produce $|A_{t,s}| = (3, 1, 1)$ matches, respectively. Then
\[
\sum_s |A_{t,s}|=5,\qquad
\sum_s |\widetilde{\mathcal T}_{t,s}|=7,\qquad
\sum_s |\widetilde{\mathcal D}_{t,s}|=10,
\]
\[
R_T(t)=5/7\approx0.71,
\qquad
R_D(t)=5/10=0.50,
\qquad
a_t=1.
\]
Thus, $71\%$ of the admitted track candidates and $50\%$ of the admitted
detection candidates participate in matches. If instead every stage has
zero admitted tracks, then
$\sum_s|\widetilde{\mathcal T}_{t,s}|=0$, so $a_t=0$; both ratios are set
to the neutral value $1$ and excluded from scoring and baseline
adaptation.

The six signals in Equation \ref{eq:features} follow from the attack's governing relation. Writing $M^{\mathrm{pre}}_t$ for the pool carried into frame $t$ and
$M^{\mathrm{post}}_t$ for the pool after births, deletions, and reconciliation,
\begin{equation}\small
M^{\mathrm{post}}_t = M^{\mathrm{pre}}_t + \Bt - X_t + C_t,
\label{eq:conservation}
\end{equation}
where $B_t$ is the number of realized births, $X_t$ is the number of
realized deletions, and $C_t$ collects any remaining change caused by
end-of-frame state reconciliation, such as duplicate removal or
tracked/lost-list merging. The current monitoring interface does not expose the deletions $X_t$ or the
composition of $C_t$ separately. What the monitor \emph{can} see is the left-hand side and the birth
\emph{demand} $U_t$ that drives $\Bt$, which is why \eqref{eq:features} tracks
$\Mt$, its increment, and $U_t$ rather than the terms it does not observe. SlowTrack
attacks the balance by inflating $U_t$ while suppressing $X_t$ (phantoms keep
being re-detected, so they never age out), so a defense must act on the two
terms it can control: admit fewer births, and delete more aggressively. The first four terms capture pool state, growth, birth pressure, and
association demand; the two ratios are intended to separate a normal
busy scene, in which real objects keep matching frame to frame, from phantom
inflation due to attacks, in which injected boxes match less often. SlowTrack re-finds some
phantoms, so the ratios fall but do not vanish: on attacked L1--L2 runs the median $R_T$ and $R_D$ were $0.52$
and $0.25$, against $0.70$ and $0.55$ on benign L1--L3 runs on the Orin.

\subsection{Robust Anomaly Detection}
\label{sec:detect}

Whether a workload is unusual depends on the scene: twelve tracks is high on an empty road and normal at a crowded intersection. TrackShield
therefore learns each signal's normal level instead of using absolute
thresholds. For each signal $i$, it maintains an exponentially weighted
moving-average (EWMA) estimate $\mu_{i,t}$ of its normal level: $\mu_{i,t}=(1-\alpha)\mu_{i,t-1}+\alpha x_{i,t}$,
where $\alpha$ controls how quickly the baseline adapts to recent
observations. It also maintains a windowed median absolute deviation
$\mathrm{MAD}_{i,t}$ to estimate the signal's normal variability. 

An initialization parameter $L_{\mathrm{init}}$ specifies the minimum
number of baseline-eligible samples required for each signal before
anomaly scoring is enabled. Until its baseline buffer contains
$L_{\mathrm{init}}$ samples, the corresponding robust score is set to
zero. Thereafter, baseline-eligible samples continuously update the EWMA
and the sliding baseline window, which retains at most $L$ samples.

A frame contributes to the baseline only if its anomaly score does not
exceed $\tau_{\mathrm{ano}}$ and no admission cap binds on it; ratio
statistics additionally require valid measurements. This limits contamination by obvious attacks.

The robust scores and aggregate anomaly score are
\begin{equation}\small
\begin{aligned}
z_i(t)&=\mathrm{clip}\!\left(
\frac{x_i(t)-\mu_i(t-1)}
     {s_i(t-1)},
-12,12\right),\\
S_t&=\sum_{\mathclap{i\in\{M,\Delta M,U,W^{\mathrm{dem}}\}}}w_i[z_i(t)]_+
+a_t\sum_{\mathclap{i\in\{R_T,R_D\}}}w_i[-z_i(t)]_+ ,
\end{aligned}
\label{eq:zscore}
\end{equation}
where $[u]_+=\max(u,0)$. Count-like signals are suspicious when they rise,
whereas match ratios are suspicious when they fall. The factor $1.4826$
provides Gaussian-consistent MAD scaling \cite{mad_hampel} within the scale
$s_i=\max\{1.4826\,\mathrm{MAD}_i+\lambda|\mathrm{med}_i|,\ \varphi_i\}+\epsilon_0$,
computed from the window before the current update,
where $\mathrm{med}_i$ is the window median, $\lambda=0.05$ is a relative
floor, $\varphi_i$ is an absolute floor ($1$ for count-like signals and
$0.02$ for ratios), and $\epsilon_0=10^{-6}$. We use MAD to reduce the influence of isolated workload bursts, such as a group entering the scene, on the scale estimate. The relative floor $\lambda|\mathrm{med}_i|$ reduces sensitivity to small changes when MAD is near zero, as for a static pool, and scales with scene density.

TrackShield distinguishes emerging and sustained anomalies. A
frame casts an anomaly vote when $S_t>\tau_{\mathrm{ano}}$, where $S_t$
is the aggregate robust anomaly score. The number of votes in the most
recent $w$ frames is
\begin{equation}\small
n_V(t)=\sum_{j=t-w+1}^{t}
\mathbf{1}[S_j>\tau_{\mathrm{ano}}].
\label{eq:nv}
\end{equation}
L1 handles an emerging anomaly: the current frame exceeds the
anomaly threshold, $S_t>\tau_{\mathrm{ano}}$, but the vote count has not
yet reached the sustained-response quorum, $n_V(t)<q$. L2 handles a
sustained anomaly and requires $n_V(t)\geq q$. With consecutive above-threshold scores, L2 engages after $q$ anomaly votes; a burst that casts fewer than $q$ votes in every $w$-frame window triggers at most L1. %

\subsection{TrackShield Levels}
\label{sec:respond}

TrackShield separates a best-effort heuristic detector from a deterministic mechanism:
L1--L2 use EWMA/MAD scores and persistence voting only to trigger cheaper, earlier responses. This detector may miss or
over-trigger; the L3 hard admission guard bounds the track and detection dimensions
admitted before cost-matrix construction, independently of  the statistical monitor. L3 is always on; it is not a severity level above L2.
\begin{tcolorbox}[breakable,colback=gray!6,colframe=gray!45,boxrule=0.4pt,
arc=1.5pt,left=4pt,right=4pt,top=3pt,bottom=3pt,
fonttitle=\bfseries\small,title={Intuition: TrackShield levels of defense}]
\small
Suppose the detector retains at most $|D|_{\max}{=}5$ boxes per frame
and lost tracks survive $R{=}4$ frames. In a benign scene with two
pedestrians, the pool holds $M_t{\approx}2$ tracks and association
touches roughly $2{\times}2{=}4$ candidate pairs per frame. Under SlowTrack, the five retained boxes are phantoms, and each box either
spawns a new track or refreshes an existing phantom track; a track that is
not refreshed expires after $R$ frames. Five such events per frame can keep
up to $R{\times}5$ tracks alive, so the pool grows toward
$R|D|_{\max}{=}20$ while the per-frame box cap is never violated;
association work grows to roughly $20{\times}5{=}100$ pairs, a
$25\times$ inflation. TrackShield controls this workload at three
points:
\begin{itemize}[leftmargin=6pt]
  \item \emph{L1 (birth capping).} Admit at most $B_{\max}$ new tracks
  per frame, ranked by detection confidence. This leaves existing tracks
  unchanged but may suppress or delay legitimate births when demand
  exceeds $B_{\max}$.
  \item \emph{L2 (track pruning).} Delete tracks flagged by at least
  two of three indicators (young age, short hit streak,
  and low confidence), but only under the sustained $q$-of-$w$ anomaly
  vote.
  \item \emph{L3 (per-stage association caps).} Before every association
  stage, admit at most $\bar K_{T,s}$ tracks and $\bar K_{D,s}$
  detections (e.g., $4{\times}4{=}16$ candidate pairs), regardless of
  how many candidates exist. This always-evaluated hard admission guard
  provides the admitted-dimension bound even if the statistical monitor
  is fooled.
\end{itemize}
\end{tcolorbox}

\paragraph*{L1 (birth capping)}
L1 admits at most $B_{\max}$ new tracks per frame, ranked by detection
confidence.

\paragraph*{L2 (track pruning)}
L2 removes suspicious tracks from the persistent pool. For each track
$j$, it computes
\begin{equation}
\small
P_j =
\mathbf{1}[\mathrm{age}_j<a_0]
+\mathbf{1}[\mathrm{hit\_streak}_j<h_0]
+\mathbf{1}[\mathrm{conf}_j<c_0],
\label{eq:phantom}
\end{equation}
where $a_0$, $h_0$, and $c_0$ are the age, hit-streak, and confidence
thresholds, respectively, and
$\mathrm{hit\_streak}_j$ is the matched-update count of track $j$ since
activation or re-activation (Section~\ref{sec:threat}); age is counted
from the track's start frame, so a newborn has age $1$. We use
$(a_0,h_0,c_0)=(3,2,0.3)$ and flag track $j$ when
$P_j\geq\gamma$, with $\gamma=2$ by default. Pruning requires at least two of three indicators: young age, a short hit streak, and low confidence. Requiring two indicators reflects the
asymmetry between retaining a phantom temporarily, which adds workload,
and deleting a real track, which can cause an identity switch.
A newborn track, however, carries both the young-age and short-hit-streak
indicators on its birth frame, whatever its confidence, and pruning runs
after birth initialization (Algorithm~\ref{alg:trackshield}). While L2 is sustained, it therefore removes every newborn track, trading
new-track admission for a smaller persistent pool. In the attacked accuracy replay at
$\tau_{\mathrm{ano}}=10$ (Section~\ref{sec:accuracy}), L2 is active on
$341$ of $400$ frames and removes all $1{,}361$ tracks born on those frames.

L2 \emph{deletes} a flagged track from both persistent lists
(\code{tracked\_stracks} and \code{lost\_stracks}), so it does not carry
forward. By contrast, the L3 guard only \emph{excludes} a track from the
current association matrix while retaining its state for possible later
admission. Deletion is more destructive: a wrongly pruned real
object may later be reinitialized with a new identity, potentially causing
an identity switch or trajectory fragmentation.
Pruning is enabled only under the sustained $q$-of-$w$ quorum, which
reduces responses to isolated anomalies.

\paragraph*{L3 (per-stage association caps)}
The L3 hard admission guard caps the inputs to each protected association stage before cost-matrix
construction. Let
\[
T_{t,s}\triangleq
|\mathcal T^{\mathrm{cand}}_{t,s}|,\qquad
D_{t,s}\triangleq
|\mathcal D^{\mathrm{cand}}_{t,s}|.
\]

The guard activates when
\begin{equation}\small
g_{t,s}={}
\mathbf{1}\!\left[
T_{t,s}>\bar K_{T,s}\ \vee\
D_{t,s}>\bar K_{D,s}\ \vee\
T_{t,s}D_{t,s}>W_{\max,s}
\right]
\label{eq:guard}
\end{equation}
and, when active, admits at most $\bar K_{T,s}$ tracks and $\bar K_{D,s}$ detections in a fixed admission order: detections by descending confidence, and tracks by descending tracklet length, then earlier start frame, then higher confidence, with the track ID as a deterministic tie-break. The indicator is evaluated on the candidate
sets entering stage $s$ immediately before its cost matrix is built. In our experiments $W_{\max,s}=\bar K_{T,s}\bar K_{D,s}$, so the product
test is redundant with the two dimension checks and the admitted work never
exceeds $\bar K_{T,s}\bar K_{D,s}$. A smaller product budget would not
tighten admission by itself, because the guard truncates only to the two
dimension caps; enforcing one would require an additional truncation rule,
which we do not evaluate.

For each association stage $s$, the dimension caps are derived from an inert benign (not under attack) profile of video frames:
\begin{equation}\small \small
\begin{aligned}
\bar K_{T,s}
  &= \left\lceil
       Q_p\!\left(
       |\mathcal{T}^{\mathrm{cand}}_{\cdot,s}|
       \right)
     \right\rceil,
\qquad
\bar K_{D,s}
  = \left\lceil
       Q_p\!\left(
       |\mathcal{D}^{\mathrm{cand}}_{\cdot,s}|
       \right)
     \right\rceil,\\
W_{\max,s}
  &= \bar K_{T,s}\bar K_{D,s},
\end{aligned}
\label{eq:capcalib}
\end{equation}
where $Q_p$ is the empirical $p$th percentile ($p=99$ by default), and $\lceil\cdot\rceil$ rounds up
to an integer count. Equation~\eqref{eq:capcalib} gives the base rule; the
latency experiments multiply the percentile by a margin of $1.01$ before
rounding (Section~\ref{sec:setup}). Track and detection dimensions are calibrated
independently because their candidate distributions differ across stages. A product-only test cannot enforce both dimension caps. For example, if
$(\bar K_{T,s},\bar K_{D,s})=(20,25)$, the candidate shape
$(T_{t,s},D_{t,s})=(5,100)$ satisfies
$T_{t,s}D_{t,s}=5\cdot100=20\cdot25=500$, although its detection dimension
is four times the configured cap. The separate condition
$D_{t,s}>\bar K_{D,s}$ correctly detects this unbalanced shape.

Detections excluded by the guard are discarded for that frame: they reach
no later stage and are not eligible for births. L3 guard-excluded tracks are carried forward rather than deleted, permitting
later readmission. In the host used for all reported ByteTrack experiments, SlowTrack's
released tracker with TrackShield hooks and its original retention rule (Tables~\ref{tab:main}--\ref{tab:accuracy} and
Figs.~\ref{fig:orin_summary}--\ref{fig:pool}), an excluded Tracked-state track stays in the Tracked state: it never moves
to the Lost set, never accrues lost time, and cannot expire after $R$
frames, and its hit streak does not change, while its chronological age
(used in~\eqref{eq:phantom}) keeps increasing. Excluded Lost-state tracks
keep accruing lost time and expire normally. In the L3-only ablation, such
tracks make up almost the entire pool (Section~\ref{sec:orinresults}).
Retaining them can preserve real trajectories, but their prediction,
ranking, reconciliation, and memory costs are outside the bound.

\begin{algorithm}[t]
\caption{TrackShield update at frame $t$. Defaults are $w=5$ and
$q=3$; $n_V(t)$ is defined in \eqref{eq:nv}.}
\label{alg:trackshield}
\footnotesize
\begin{algorithmic}[1]
\setlength{\itemsep}{0pt}
\setlength{\parsep}{0pt}
\setlength{\topsep}{0pt}

\STATE $M_t\gets|\mathcal T_{t-1}|$;
       $\Delta M_t\gets M_t-M_{t-1}$.
\FOR{each protected association stage $s$}
  \STATE Form
  $(\mathcal T^{\mathrm{cand}}_{t,s},
    \mathcal D^{\mathrm{cand}}_{t,s})$.
  \STATE Evaluate $g_{t,s}$ using the stage-local candidate dimensions.
  \IF{$g_{t,s}=1$}
    \STATE Truncate to
    $|\mathcal T_{t,s}|\leq\bar K_{T,s}$ and
    $|\mathcal D_{t,s}|\leq\bar K_{D,s}$.
  \ENDIF
  \STATE Associate
  $(\mathcal T_{t,s},\mathcal D_{t,s})$;
  log pre/post-cap dimensions.
\ENDFOR
\STATE Set $U_t$ to the number of detections left unmatched after the
final association stage (before ByteTrack's \code{det\_thresh} score filter
and before L1 capping).
\STATE Form
$x_t=(M_t,\Delta M_t,U_t,W_t^{\mathrm{dem}},R_T,R_D)$
and compute $S_t$ and $n_V(t)$.
\IF{$n_V(t)\geq q$}
  \STATE $\ell_t\gets2$ \COMMENT{sustained}
\ELSIF{$S_t>\tau_{\mathrm{ano}}$}
  \STATE $\ell_t\gets1$ \COMMENT{emerging}
\ELSE
  \STATE $\ell_t\gets0$ \COMMENT{normal}
\ENDIF
\IF{$\ell_t\geq1$}
  \STATE Keep at most $B_{\max}$ birth candidates.
\ENDIF
\STATE Initialize admitted births and reconcile tracked, lost, and
removed sets.
\IF{$\ell_t=2$}
  \STATE Prune tracks satisfying the pruning rule~\eqref{eq:phantom}.
\ENDIF
\end{algorithmic}
\end{algorithm}

\paragraph*{Complexity}
The reported tracking-stage latency $T_{\mathrm{track}}^{\mathrm{TS}}$
is the enclosing tracker-update time and includes all TrackShield
operations. Its cost can be decomposed conceptually as
\begin{equation}\small
\small
\begin{aligned}
T_{\mathrm{track}}^{\mathrm{TS}}
=
T_{\mathrm{pred}}
+T_{\mathrm{L3}}
+\sum_{s=1}^{S}
T_{\mathrm{assoc}}(K_{T,s},K_{D,s})
+T_{\mathrm{monitor}} \\
+T_{\mathrm{L1/L2}}
+T_{\mathrm{recon}},
\label{eq:cost-decomp}
\end{aligned}
\end{equation}
For ByteTrack's fixed $S=3$ stages, signal formation, scoring, voting,
and guard checks are $O(1)$ per frame, apart from $O(L)$ median/MAD
maintenance per accepted baseline update ($L=120$). L1 sorts $U_t$ birth
candidates, $O(U_t\log U_t)$; L2 scans the persistent pool, $O(M_t)$; and
the L3 guard sorts the uncapped stage inputs $N_{T,s}$ and $N_{D,s}$
before truncation, $O(\sum_{s}[N_{T,s}\log N_{T,s}+N_{D,s}\log N_{D,s}])$,
which linear-time top-$K$ selection could reduce.
Proposition~\ref{prop:cap} therefore bounds only the post-cap dimensions
$(K_{T,s},K_{D,s})$ governing the association terms; guard selection,
prediction, monitoring and response actions, reconciliation, and
persistent-state maintenance may still depend on uncapped state.

\begin{proposition}[Conditional association-admission bound]
\label{prop:cap}
Suppose that at every protected association stage $s\in\{1,\dots,S\}$ the
hard admission guard is applied \emph{before} cost-matrix construction, admitting at
most $\bar K_{T,s}$ tracks and $\bar K_{D,s}$ detections (the admission order determines \emph{which} candidates are
admitted, never more than $\bar K_{T,s}$ tracks and $\bar K_{D,s}$ detections). Then for any finite candidate sets (in particular under an adversarial flood of arbitrarily many detections or births), $K_{T,s}\le\bar K_{T,s}$, $K_{D,s}\le\bar K_{D,s}$,
and
\[
  \Wt^{\mathrm{adm}} = \sum_{s=1}^{S} K_{T,s}K_{D,s}
  \;\le\; \sum_{s=1}^{S} \bar K_{T,s}\bar K_{D,s}.
\]
Consequently, any assignment implementation whose worst-case operation
count is monotone non-decreasing in both admitted dimensions has a finite,
implementation-specific worst-case operation count at those dimensions.
\end{proposition}

\begin{proof}
The guard enforces $K_{T,s}\leq\bar K_{T,s}$ and
$K_{D,s}\leq\bar K_{D,s}$ before the cost matrix is built. Multiplying
these inequalities gives the stage-wise bound, and summing over the
protected stages gives the frame-wise bound; monotonicity of the solver's worst-case operation count in both admitted
dimensions gives the operation-count bound. 
\end{proof}

\begin{table}[htbp]
\caption{TrackShield parameters.}
\label{tab:parameters}
\centering
\scriptsize
\setlength{\tabcolsep}{2pt}
\begin{tabularx}{\columnwidth}{@{}
    >{\raggedright\arraybackslash}p{0.15\columnwidth}
    >{\raggedright\arraybackslash}p{0.36\columnwidth}
    >{\raggedright\arraybackslash}X
@{}}
\toprule
\textbf{Parameter}
& \textbf{Role}
& \textbf{Effect of increase} \\
\midrule
$\bar K_{T,s}$, $\bar K_{D,s}$
& Configured track/detection dimension caps
& More candidates and work admitted \\
$W_{\max,s}$
& Stage-local product-trigger threshold
& Less frequent product-triggered capping \\
$\tau_{\mathrm{ano}}$
& Anomaly-score threshold
& Later, less frequent heuristic response \\
$w$
& Anomaly-vote window length
& Longer vote history \\
$q$
& Sustained-response quorum
& More anomaly votes required \\
$B_{\max}$
& Per-frame L1 birth cap
& More births admitted when L1 is active \\
$\gamma$
& L2 pruning threshold
& Fewer tracks satisfy the pruning rule \\
$a_0,h_0,c_0$
& Age, hit streak, and confidence thresholds
& Changes pruning selectivity when L2 is active\\
$\alpha$
& EWMA adaptation rate
& Faster baseline adaptation \\
$L$
& Maximum baseline-window length
& More accepted history retained \\
$L_{\mathrm{init}}$
& Baseline samples required to enable scoring
& Later score activation \\
\bottomrule
\end{tabularx}
\end{table}

The dimension caps in Table~\ref{tab:parameters} control the admitted-workload--recall tradeoff;
$B_{\max}$ and $\tau_{\mathrm{ano}}$ control heuristic birth admission
and response sensitivity. Tight caps may truncate legitimate candidates,
while aggressive anomaly thresholds may cause excessive pruning.

\section{Performance Evaluation}
\label{sec:eval}
\subsection{Experimental Setup}
\label{sec:setup}
\begin{figure*}[!t]
\centering
\includegraphics[width=0.98\textwidth]{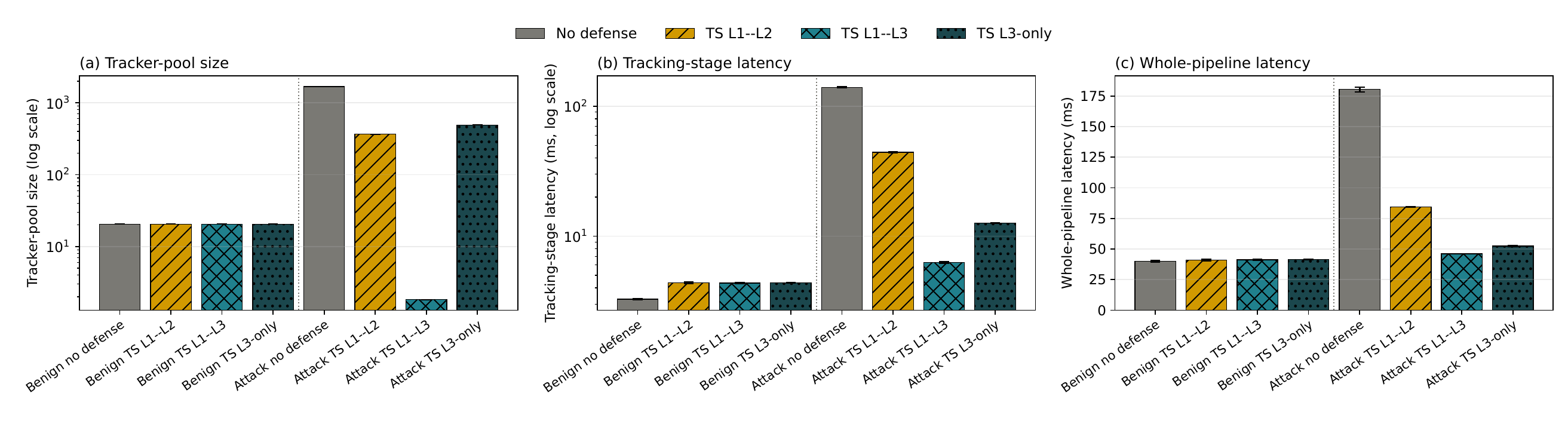}
\caption{Pool and latency results for the eight MOT17-02 conditions in
Table~\ref{tab:main}. Bars show medians and error bars show IQRs across
five repetition-level means.}
\label{fig:orin_summary}
\end{figure*}

\textbf{Platform and configuration.}
Latency experiments use YOLOX-S+ByteTrack at $608\times1088$, FP16, and at
most 500 retained detections per frame on a Jetson AGX Orin (MAXN, locked
clocks), with detector confidence/NMS thresholds $0.25/0.45$,
\code{track\_thresh}$=0.6$, \code{match\_thresh}$=0.9$, and $R=30$. All
ByteTrack results use SlowTrack's released tracker~\cite{slowtrack}, against
which the attack is defined. It differs from official ByteTrack in its
evaluation preprocessing (RGB/255, without the ImageNet mean/std
normalization of official ByteTrack) and in its second-pass detection set:
an indexing quirk feeds the second stage the $\min(m,100)$ highest-scoring
detections overall ($m$ being the middle-band count) instead of the middle
band, so a detection matched in the first stage can be reused there. Adversarial frames use the confidence-maximization path of the released
SlowTrack code (220 normalized-gradient steps per frame pair, no norm
budget), with its input normalization removed to match this preprocessing;
the box-area and feature-matching terms of Section~\ref{sec:slowtrack} are
inactive in this path.
Accuracy experiments replay cached detections (confidence $0.001$, NMS
$0.7$, \code{track\_thresh}$=0.5$, \code{match\_thresh}$=0.8$) over the matched 400-frame range, separately from the latency runs, and are scored with
\texttt{motmetrics} against MOTChallenge ground truth. Caps use benign $p=99$ calibration on clean frames of the same sequence.
For the MOT17-02 latency runs, calibrated on frames 101--400 with a $1.01$ margin, the caps are $(24,20)$,
$(4,19)$, and $(2,4)$ for the high-score, low-score, and unconfirmed
stages, so by Proposition~\ref{prop:cap} at most
$24{\cdot}20+4{\cdot}19+2{\cdot}4=564$ candidate pairs are admitted to association per frame. The accuracy replay calibrates on the full 600-frame benign sequence without a margin, giving $(37,25)$, $(13,81)$, and $(2,4)$. The monitor baseline is warm-started from the inert
profiling pass.

\textbf{Defense configuration.}
TrackShield protects all three association stages and the birth and
reconciliation steps. Unless stated otherwise, $\tau_{\mathrm{ano}}=6$,
$w=5$, $q=3$, $\alpha=0.02$, $B_{\max}=15$, $L=120$, $L_{\mathrm{init}}=20$,
and $w_{R_D}=0.5$ (all other weights $1$); the accuracy sweep evaluates
$\tau_{\mathrm{ano}}\in\{4,6,8,10\}$.

\textbf{Protocol.}
Paired runs use identical frames and detector settings. Pipeline latency
is CUDA-synchronized whole-frame timing over the last 300 of 400 frames;
each ByteTrack Jetson condition is repeated five times, and
Tables~\ref{tab:main} and~\ref{tab:crossseq} report medians of
repetition-level means. R-Lat is mean pipeline latency relative to the
paired benign no-defense baseline minus one, and $\overline{M_t}$ is the
mean pool size. Cross-sequence attacks are generated per sequence; the
BoT-SORT study reuses the MOT17-02 attack without re-optimization.

\subsection{Edge Latency and Workload}
\label{sec:orinresults}

The Orin campaign uses $\tau_{\mathrm{ano}}=6$.
Table~\ref{tab:main} reports the eight evaluated conditions, and
Fig.~\ref{fig:orin_summary} visualizes their pool and latency results.

Under attack, the undefended pipeline reaches
$R\text{-}\mathrm{Lat}=3.48$, a mean persistent pool of $1{,}677.7$
tracks, and mean tracking-stage latency of $139.90$\,ms. Full
TrackShield reduces these values to $0.15$, $1.8$ tracks, and
$6.27$\,ms, respectively. At least one calibrated per-stage admission
cap binds on every attacked frame. Under benign input, the mean pool
remains $20.5$ tracks across all conditions, while tracking-stage latency
ranges from $3.27$ to $4.38$\,ms.

For the MOT17-02-FRCNN campaign in Table~\ref{tab:main}, the ablations
reveal mechanism interaction. Without L3, attacked L1--L2 drops $414{,}395$ birth candidates over the $5\times300$ measured frame-repetitions (a sum of per-frame counts, not distinct objects); together with L2's removal of newborn tracks, this reduces the mean pool from $1{,}677.7$ to $365.0$ tracks.
With L3 enabled, L1 drops no birth candidates under full L1--L3, because births come only from the at most $4$ detections admitted to the unconfirmed stage, below $B_{\max}=15$; L2 remains active,
recording $5{,}335$ per-frame pruning events. With L1 inactive, the
L3-only--to--full-policy difference is attributable to L2 in this
campaign: the mean pool falls from $489.7$ to $1.8$ tracks and
tracking-stage latency from $12.62$ to $6.27$\,ms. The L3-only pool, almost entirely guard-excluded Tracked-state tracks (Section~\ref{sec:respond}), grows linearly over the measured window, so its mean depends on the window length.

TrackShield's internal overhead is approximately $1$\,ms per frame. Its
timers cover monitoring, birth capping, pruning, association-input
capping, and bookkeeping; the median across five repetition-level mean
timer sums is $0.99$\,ms under benign input and $1.20$\,ms under attack.
This overhead is included in the tracking-stage latency of Table~\ref{tab:main}. The whole-pipeline effect is small (benign
P99 $41.4$ versus $42.4$\,ms), but relative to the benign tracker it is not:
benign tracking-stage latency rises from $3.27$ to $4.35$\,ms under full
L1--L3, about $33\%$.

\begin{figure}[!h]
\centering
\includegraphics[width=0.85\columnwidth]{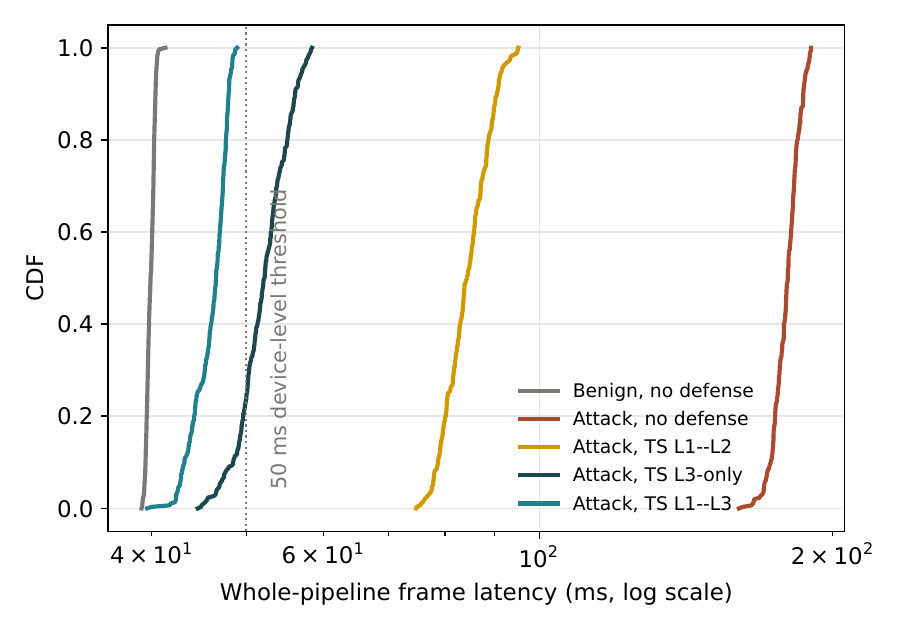}
\caption{Whole-pipeline per-frame latency CDFs on the Jetson AGX Orin for
the benign no-defense baseline and the four attacked conditions; each CDF
pools the per-frame latencies of all five repetitions.  The
vertical line marks the 50\,ms device-level threshold used for the
deadline-miss analysis.}
\label{fig:orincdf}
\end{figure}

\begin{figure*}[!t]
\centering
\includegraphics[width=0.9\linewidth]{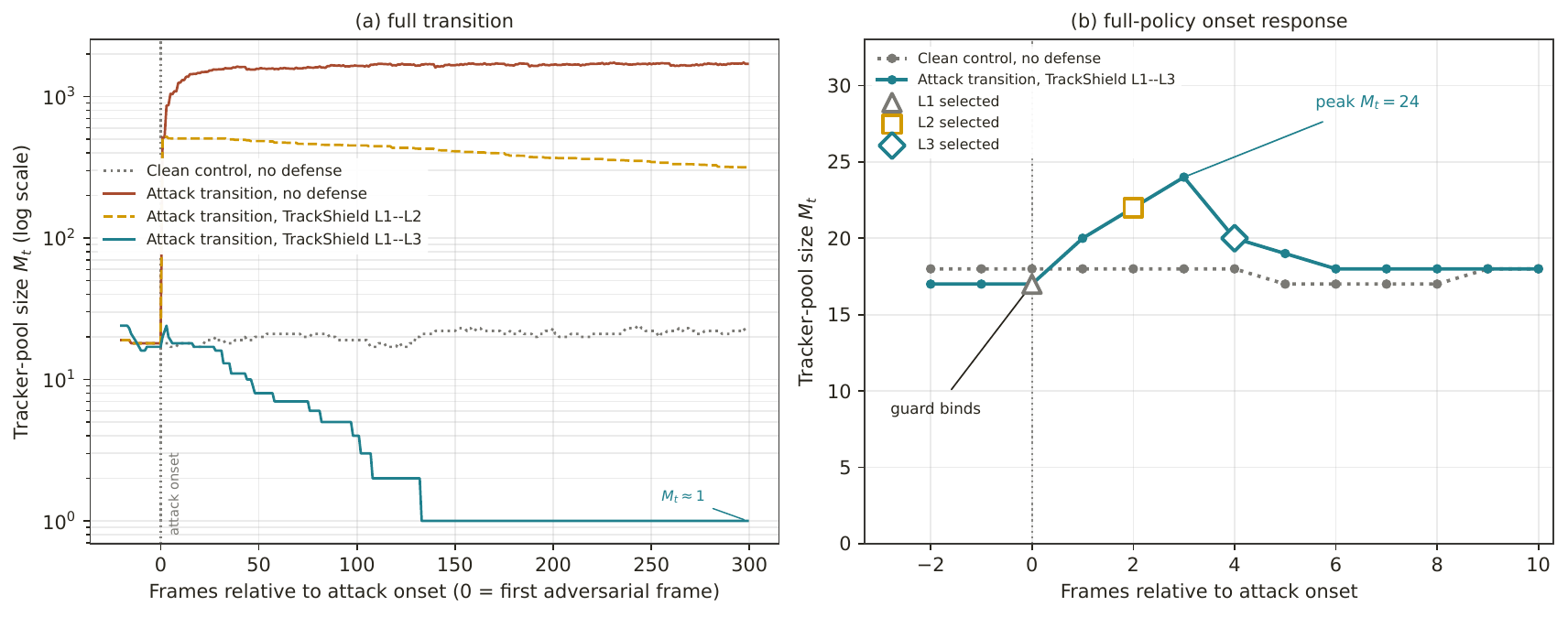}
\caption{Persistent tracker-pool response during a separate
clean-to-adversarial transition run on Jetson AGX Orin. Time is
relative to the first adversarial frame ($t=0$); panel~(a) shows the full
trajectory on a log scale, and panel~(b) enlarges the onset of the defended (L1--L3) run on a linear scale. $M_t$ is sampled at frame entry, so the first attack-induced change appears at $t=1$. }
\label{fig:pool}
\end{figure*}

\begin{table}[!t]
\caption{Jetson AGX Orin results on MOT17-02-FRCNN. R-Lat, pool mean,
and tracking mean are medians of five repetition-level means; P99 pools
all post-warm-up frames. L1--L2 disables association capping, whereas
L3-only disables birth capping and pruning.}
\label{tab:main}
\centering
\footnotesize
\setlength{\tabcolsep}{2.7pt}
\begin{tabular}{L{2.3cm} cccc}
\toprule
\textbf{Condition} &
\textbf{$R$-Lat} &
\textbf{Pool mean} &
\textbf{Trk mean (ms)} &
\textbf{Pipe P99 (ms)} \\
\midrule
Benign, no defense & $0.00$ & $20.5$ & $3.27$ & $41.4$ \\
Benign, TS L1--L2  & $0.03$ & $20.5$ & $4.38$ & $42.5$ \\
Benign, TS L1--L3  & $0.02$ & $20.5$ & $4.35$ & $42.4$ \\
Benign, TS L3-only & $0.03$ & $20.5$ & $4.38$ & $42.5$ \\
\midrule
Attack, no defense & $3.48$ & $1{,}677.7$ & $139.90$ & $190.9$ \\
Attack, TS L1--L2  & $1.09$ & $365.0$ & $44.18$ & $94.6$ \\
Attack, TS L1--L3  & $0.15$ & $1.8$ & $6.27$ & $48.8$ \\
Attack, TS L3-only & $0.30$ & $489.7$ & $12.62$ & $58.4$ \\
\bottomrule
\end{tabular}
\end{table}

The object detector-dominated benign pipeline already exceeds the commonly used
33\,ms reference threshold, so that threshold does not distinguish
benign operation from an attacked pipeline. We therefore evaluate
deadline misses at an attainable 50\,ms device-level threshold. The miss
rate is the fraction of pooled post-warm-up frames across all five
repetitions whose whole-frame latency exceeds 50\,ms. The undefended attack and L1--L2 exceed this threshold on 100\% of frames. L3-only, a static per-stage cap without monitoring or pruning, reduces the miss rate to $79.6\%$, whereas full L1--L3 reduces it to
$0.7\%$, compared with $0.0\%$ for the benign conditions. Fig.~\ref{fig:orincdf} shows the full per-frame latency distributions.

The separate transition run in Fig.~\ref{fig:pool} reveals onset dynamics
hidden by the post-warm-up aggregates in Table~\ref{tab:main}.  In the undefended run, the pool rises rapidly toward
approximately $1{,}700$ tracks. Under full L1--L3, at least one
stage-local admission cap binds on the first adversarial frame, and the
pool peaks briefly at $24$ tracks before falling below its benign level.
The benign tracks present before onset gradually stop matching and expire $R=30$ frames after their last match; with L2 removing newborns, the defended pool declines stepwise to $1$--$2$ tracks within about 130 frames, below the benign level of about $20$. This small pool is a result of the full policy, not of
Proposition~\ref{prop:cap}.

\subsection{Cross-Sequence Robustness}
\label{sec:crossseq}
We apply the protocol of Section~\ref{sec:setup} to seven MOT17 sequences
(02, 04, 05, 09, 10, 11, and 13), generating attacks independently per sequence.

Before interpreting defended results as mitigation evidence, we apply a
predeclared \emph{attack-efficacy gate}: under no defense, the ratio of
the attacked mean persistent-pool size to its benign counterpart must be
at least $1.5\times$. A sequence that fails this gate is retained as a
negative or specificity result but excluded from mitigation aggregates.
In Table~\ref{tab:crossseq}, R-Lat is residual whole-pipeline latency
inflation; ``Trk red.'' is the undefended attacked tracking-stage latency
divided by its full-L1--L3 counterpart; and ``Def./benign'' is the
defended attacked mean pool divided by the benign no-defense mean pool.
The six-sequence mean is unweighted; ratios are computed from unrounded measurements.

\begin{table}[!hbt]
\caption{Cross-sequence Jetson AGX Orin results over 300 measured frames
after 100 warm-up frames. The MOT17-02 row is an independent rerun of
Table~\ref{tab:main}'s protocol. $\dagger$ marks a gate-failing sequence
excluded from the six-sequence mean.}
\label{tab:crossseq}
\centering
\scriptsize
\setlength{\tabcolsep}{1.8pt}
\begin{tabular}{l r rr r rr r}
\toprule
\textbf{Sequence} &
\textbf{Pool} &
\multicolumn{2}{c}{\textbf{Attack R-Lat}} &
\textbf{Trk} &
\textbf{Benign} &
\textbf{Def.\ atk.} &
\textbf{Def./} \\
\cmidrule(lr){3-4}
&
\textbf{ratio} &
\textbf{no def.} &
\textbf{L1--L3} &
\textbf{red.} &
\textbf{pool} &
\textbf{pool} &
\textbf{benign} \\
\midrule
MOT17-02
& $82.0\times$ & $3.51$ & $0.146$ & $22.9\times$
& $20.5$ & $1.83$ & $0.089$ \\
MOT17-05
& $159.9\times$ & $3.96$ & $0.173$ & $24.7\times$
& $11.1$ & $1.73$ & $0.156$ \\
MOT17-09
& $198.1\times$ & $4.02$ & $0.143$ & $26.9\times$
& $9.7$ & $1.43$ & $0.146$ \\
MOT17-10
& $100.3\times$ & $3.82$ & $0.128$ & $24.6\times$
& $18.2$ & $2.30$ & $0.127$ \\
MOT17-11
& $200.8\times$ & $4.12$ & $0.192$ & $26.2\times$
& $9.6$ & $1.88$ & $0.195$ \\
MOT17-13
& $103.0\times$ & $3.75$ & $0.121$ & $24.5\times$
& $17.3$ & $1.81$ & $0.105$ \\
\midrule
\textbf{6-seq mean}
& $\mathbf{140.7\times}$
& $\mathbf{3.86}$
& $\mathbf{0.150}$
& $\mathbf{25.0\times}$
& $\mathbf{14.4}$
& $\mathbf{1.83}$
& $\mathbf{0.136}$ \\
\midrule
MOT17-04$^\dagger$
& $1.2\times$ & $0.06$ & $-0.008$ & $1.8\times$
& $41.5$ & $4.00$ & $0.096$ \\
\bottomrule
\end{tabular}
\end{table}

On the six sequences that pass the efficacy gate, the undefended attack
inflates the mean persistent pool by $82$--$201\times$ relative to benign
operation, with an unweighted mean of $140.7\times$. The corresponding
undefended R-Lat values range from $3.51$ to $4.12$. Full L1--L3 reduces
tracking-stage latency by $22.9$--$26.9\times$, with an unweighted mean
reduction of $25.0\times$, and lowers R-Lat to $0.121$--$0.192$, with an
unweighted mean of $0.150$. The calibrated admission caps bind on every
attacked frame in all six attack-effective sequences.

The defended attacked pools are $8.9\%$--$19.5\%$ of the corresponding
benign pools, below the benign level because L2 also removes legitimate
newborns (Section~\ref{sec:respond}).

\paragraph*{MOT17-04 negative case}
MOT17-04 fails the $1.5\times$ attack-efficacy gate: its attacked mean
pool of $50.2$ tracks is only $1.2\times$ its benign $41.5$. Full L1--L3 reduces its attacked pool to $4.00$ tracks, $9.6\%$ of the benign pool, which suggests possible over-suppression.

\subsection{Tracking Utility}
\label{sec:accuracy}

\begin{table}[!t]
\caption{Tracking accuracy on MOT17-02-FRCNN over a matched 400-frame
scope; benign and attacked rows are scored against the same 400-frame
ground truth ($12{,}075$ boxes). Multiple Object Tracking Accuracy (MOTA) and identity F1 score
(IDF1) are percentages; IDSW denotes identity switches, FP false
positives, and FN false negatives. Interv.\ is the percentage of frames
with selected response level $\ell_t>0$, and Bind is the percentage on
which at least one per-stage admission cap truncates a candidate set; Bind exceeds Interv.\ in every TrackShield row because the guard is evaluated independently of the monitor.
Dashes indicate that TrackShield-specific rates are not applicable.}
\label{tab:accuracy}
\centering
\scriptsize
\setlength{\tabcolsep}{2.0pt}
\begin{tabular}{l ccccc cc}
\toprule
\textbf{Condition} & \textbf{MOTA} & \textbf{IDF1} &
\textbf{IDSW} & \textbf{FP} & \textbf{FN} &
\textbf{Interv.} & \textbf{Bind} \\
\midrule
\multicolumn{8}{l}{\emph{Benign input, matched 400-frame range}}\\
No defense
& $57.57$ & $56.83$ & $89$ & $1{,}175$ & $3{,}860$ & --- & --- \\
TS L1--L3, $\tau{=}4$
& $17.43$ & $25.01$ & $0$ & $32$ & $9{,}938$ & $99.8\%$ & $100.0\%$ \\
TS L1--L3, $\tau{=}6$
& $56.51$ & $56.36$ & $83$ & $1{,}162$ & $4{,}007$ & $5.2\%$ & $11.2\%$ \\
TS L1--L3, $\tau{=}8$
& $56.75$ & $56.32$ & $85$ & $1{,}162$ & $3{,}975$ & $0.5\%$ & $12.2\%$ \\
TS L1--L3, $\tau{=}10$
& $56.75$ & $56.32$ & $85$ & $1{,}162$ & $3{,}975$ & $0.5\%$ & $12.2\%$ \\
\midrule
\multicolumn{8}{l}{\emph{Attacked input, same 400-frame range}}\\
No defense
& $-1{,}544.08$ & $1.12$ & $752$ & $189{,}508$ & $8{,}263$ & --- & --- \\
TS L1--L3, $\tau{=}4$
& $-0.07$ & $0.43$ & $0$ & $36$ & $12{,}047$ & $94.8\%$ & $100.0\%$ \\
TS L1--L3, $\tau{=}6$
& $-0.07$ & $0.43$ & $0$ & $36$ & $12{,}047$ & $94.8\%$ & $100.0\%$ \\
TS L1--L3, $\tau{=}8$
& $-0.13$ & $0.43$ & $0$ & $44$ & $12{,}047$ & $91.8\%$ & $100.0\%$ \\
TS L1--L3, $\tau{=}10$
& $-0.13$ & $0.43$ & $0$ & $44$ & $12{,}047$ & $91.2\%$ & $100.0\%$ \\
\bottomrule
\end{tabular}
\end{table}

Table~\ref{tab:accuracy} evaluates tracking utility under benign and attacked input over the same 400-frame range.

Under benign input at the default $\tau_{\mathrm{ano}}=6$, MOTA decreases
from $57.57$ to $56.51$ ($-1.06$ points) and IDF1 from $56.83$ to $56.36$
($-0.47$ points); at $\tau_{\mathrm{ano}}=10$ the losses are $0.82$ and
$0.51$ points, so neither setting is better on both metrics. By contrast, at $\tau_{\mathrm{ano}}=4$, the heuristic
response is active on nearly every benign frame and an admission cap binds
on every frame. Because sustained L2 removes every newborn track, FP falls
to 32 and IDSW to zero, but FN rises to $9{,}938$, causing MOTA to drop to $17.43\%$. In an auxiliary
accuracy-replay sensitivity test, not used in any latency result,
exempting each newborn on its birth frame and lowering $h_0$ to $1$ removes this collapse (MOTA $56.99$
at $\tau_{\mathrm{ano}}=4$) and leaves the other benign rows unchanged or
slightly improved. Among the tested settings, $\tau_{\mathrm{ano}}=6$ is
the lowest threshold that avoids the collapse under the default rule. At $\tau_{\mathrm{ano}}=8$ and $10$, the selected response level is identical on every benign frame, hence the identical benign rows.

Without defense, SlowTrack reduces MOTA/IDF1 from $57.57/56.83$ under
benign input to $-1{,}544.08/1.12$ under attack: the phantom detections
become $189{,}508$ false-positive outputs, while FN rises from $3{,}860$ to
$8{,}263$. The attack degrades tracking accuracy as well as increasing workload. At
$\tau_{\mathrm{ano}}=10$, full TrackShield removes nearly all of these false positives ($44$ versus
$189{,}508$), while FN rises to $12{,}047$ of the $12{,}075$ ground-truth
boxes, giving attacked MOTA/IDF1 of $-0.13/0.43$. In the same auxiliary test (birth-frame exemption and $h_0=1$) at $\tau_{\mathrm{ano}}=10$, the replayed attacked pool rises from $2.3$ to $50.9$ tracks and the false positives to $6{,}021$ (MOTA $-45.84$), while FN stays above $95\%$.

\subsection{Cross-Tracker Admission-Control Case Study}
\label{sec:crosstracker}
The preceding evaluations use ByteTrack. To examine the admission-control
mechanism beyond ByteTrack, we apply the MOT17-02 attack, with the same
cached-detection protocol, to a tracker-specific BoT-SORT-ReID
integration. This integration retains per-stage association caps and track pruning, but has no L1 birth cap and does not implement the
complete ByteTrack monitoring interface; Table~\ref{tab:botsort} reports this integration as a case study. Pool, accuracy, and cap-binding values are medians over five GPU-cluster
repetitions. The Trk latency row (including ReID) is the mean over all 400
frames of one Jetson run per condition, with no run-to-run spread, and is not comparable with Table~\ref{tab:main}. The cluster runs use caps $(30,20)$, $(12,82)$, and $(2,4)$, calibrated on their 400-frame benign profile without a margin; the Jetson run uses $(29,20)$, $(13,82)$, and $(2,4)$ from its own benign profile. The integration's monitor signals differ from Algorithm~\ref{alg:trackshield}: $A_t$ counts activated tracks and $U_t$ unconfirmed tracks, so $R_T$ and $R_D$ are not confined to $[0,1]$, and monitor and pruning behavior is not directly comparable with ByteTrack.

\begin{table}[!t]
\caption{BoT-SORT-ReID admission-control case study on MOT17-02-FRCNN.
The integration uses per-stage caps and pruning, with tracker-specific
monitor signals and no L1 birth cap. Accuracy metrics use the same
400-frame ground truth as Table~\ref{tab:accuracy}. IDSW denotes identity
switches.}
\label{tab:botsort}
\centering
\scriptsize
\setlength{\tabcolsep}{2.5pt}
\begin{tabular}{@{}lrrrr@{}}
\toprule
& \multicolumn{2}{c}{\textbf{Benign}} &
  \multicolumn{2}{c}{\textbf{Attack}} \\
\cmidrule(lr){2-3}\cmidrule(lr){4-5}
\textbf{Metric} &
\textbf{No def.} & \textbf{TrackShield} &
\textbf{No def.} & \textbf{TrackShield} \\
\midrule
Mean pool & $22.60$ & $22.66$ & $3{,}975.29$ & $1.69$ \\
MOTA      & $47.41$ & $46.82$ & $-6{,}972.63$ & $-0.81$ \\
IDF1      & $56.09$ & $56.66$ & $0.42$ & $0.21$ \\
IDSW      & $80$ & $81$ & $1{,}551$ & $2$ \\
FP        & $980$ & $1{,}027$ & $847{,}624$ & $113$ \\
FN        & $5{,}290$ & $5{,}313$ & $4{,}845$ & $12{,}058$ \\
Bind      & $0\%$ & $6.75\%$ & $0\%$ & $100\%$ \\
\midrule
Trk latency & $108.6$\,ms & $110.3$\,ms & $26.65$\,s & $10.94$\,s \\
\bottomrule
\end{tabular}
\end{table}

On BoT-SORT-ReID, the transferred attack passes the efficacy gate and
TrackShield reduces the mean attacked pool from 3,975.29 to 1.69,
with caps binding on every attacked frame. Benign MOTA decreases by $0.59$ points, while benign IDF1 rises slightly,
from $56.09$ to $56.66$; under attack, FN reaches $12{,}058$ of $12{,}075$. TrackShield reduces attacked
tracking-stage latency from 26.65\,s to 10.94\,s on Jetson, but ReID alone accounts
for approximately 10.85\,s because it processes the raw detections before
the admission guard. The guard enforces Proposition~\ref{prop:cap}, but upstream ReID cost lies outside it.

\subsection{Limitations}
\label{sec:limitations}
Proposition~\ref{prop:cap} covers only admitted association work (Section~\ref{sec:respond}), not object detector, ReID, or end-to-end cost. The evaluated attack is SlowTrack's confidence-maximization path, white-box
and without a perturbation budget, transferred from the undefended
pipeline. An adaptive attacker cannot raise admitted association work above $\sum_s\bar K_{T,s}\bar K_{D,s}$ (Proposition~\ref{prop:cap}); we have not evaluated small-norm perturbations or attacks that evade anomaly voting or target uncapped stages such as ReID, prediction, and reconciliation.

\section{Conclusion}
\label{sec:conclusion}
TrackShield addresses SlowTrack-style availability attacks at the layer
where their workload accumulates: persistent tracker state and data
association. The always-on L3 guard provides a deterministic bound on
the track and detection dimensions admitted to each protected association
stage, while L1--L2 limit
pool growth through birth capping and track pruning. On a Jetson AGX Orin, full TrackShield reduces attacked R-Lat from $3.48$ to $0.15$ on MOT17-02, and attacked tracking-stage latency by $22.9$--$26.9\times$ across six attack-effective MOT17 sequences. Under attack, adversarial detections corrupt the tracker's input, and TrackShield prioritizes timely output over tracking accuracy. Timely, if degraded, output still gives downstream monitoring and fallback mechanisms an opportunity to respond, whereas a stalled or unresponsive pipeline leaves downstream modules acting on stale world state. Remaining challenges include restoring tracking utility under attack, addressing
pipeline costs outside the association bound, and evaluating and defending against adaptive attacks.
\bibliographystyle{IEEEtran}
\bibliography{references}
\end{document}